\documentclass[conference]{IEEEtran}
\IEEEoverridecommandlockouts
 \usepackage{lipsum,graphicx,multicol}
 \usepackage{comment}
\usepackage{color,soul}
\usepackage{mathtools}
\usepackage{graphicx}
\usepackage{color,soul}
\usepackage{setspace}
\usepackage{multirow}
\usepackage{tablefootnote}
\usepackage[norelsize, linesnumbered, ruled, lined, boxed, commentsnumbered]{algorithm2e}
\usepackage[normalem]{ulem}
\usepackage{comment}
\usepackage[font=small,labelfont=bf]{caption}
\usepackage{subcaption}
\usepackage{bm, amsmath, amssymb, amsthm}
\usepackage{amsfonts}
\usepackage {amssymb,amsmath,xcolor}
\usepackage{booktabs}
\usepackage{cite}
\usepackage{array}
\usepackage[acronym]{glossaries}
\usepackage{cases}
\usepackage{cleveref}
\usepackage{algorithmic}

\usepackage{amsthm}
\usepackage{array}

\def\BibTeX{{\rm B\kern-.05em{\sc i\kern-.025em b}\kern-.08em
    T\kern-.1667em\lower.7ex\hbox{E}\kern-.125emX}}
    
\newtheorem{problem}{Problem}

\newtheorem{proposition}{Proposition}

\newcommand{\bieee}{\begin{IEEEeqnarray}{rCl}}
\newcommand{\eieee}{\end{IEEEeqnarray}}

\usepackage{amsthm,mdframed,calc}
 \usepackage{float}
\usepackage{diagbox}
\SetKwInput{KwInput}{Input}                
\SetKwInput{KwOutput}{Output}              

\begin{document}

\title{Balancing Privacy and Robustness in Coded Computing Under Profiled Workers}

\author{\IEEEauthorblockN{Rimpi Borah, J. Harshan and Aaditya Sharma}
\IEEEauthorblockA{Department of Electrical Engineering, Indian Institute of Technology Delhi, India}}

\maketitle

\begin{abstract}
In distributed computing with untrusted workers, the assignment of evaluation indices plays a critical role in determining both privacy and robustness. In this work, we study how the placement of unreliable workers within the Numerically Stable Lagrange Coded Computing (NS-LCC) framework influences privacy and the ability to localize Byzantine errors. We derive analytical bounds that quantify how different evaluation-index assignments affect privacy against colluding curious workers and robustness against Byzantine corruption under finite-precision arithmetic. Using these bounds, we formulate optimization problems that identify privacy-optimal and robustness-optimal index placements and show that the resulting assignments are fundamentally different. This exposes that index choices that maximizes privacy degrade error-localization, and vice versa. To jointly navigate this trade-off, we propose a low-complexity greedy assignment strategy that closely approximates the optimal balance between privacy and robustness.

\end{abstract}

\section{Introduction}
In a typical distributed coded computing framework, the dataset is dispersed across a network of servers, enabling decentralized and parallel computation. While such frameworks offer significant advantages when deployed over trusted servers, they pose serious challenges in environments that include untrusted servers. These untrusted servers may be Byzantine, who can introduce incorrect computations, or curious, who may attempt to infer information about the dataset. Therefore, robustness against such adversarial behaviors is an important design consideration. Lagrange Coded Computing (LCC) is one such framework that provides all these features \cite{b3}. However, LCC operates over finite fields, whereas most real-world computations are performed on floating-point data. Although Analog LCC (ALCC) \cite{b6} operates on floating-point data, its decoding requires inverting a Vandermonde matrix whose condition number grows exponentially with the number of stragglers (i.e., slow or unresponsive workers), causing numerical instability. In contrast, Numerically Stable LCC (NS-LCC) \cite{f1} retains Lagrange-based encoding, however uses Chebyshev nodes for evaluation, leading to the inversion of a Chebyshev Vandermonde matrix whose condition number grows polynomially with stragglers, thereby ensuring numerical stability.

Although the ALCC framework in \cite{b6} accounts for stragglers as well as \textit{curious} and \textit{Byzantine} workers \cite{a2}, the NS-LCC framework primarily addresses only stragglers and does not explicitly consider untrusted workers. Such untrusted workers commonly arise in peer-to-peer networks or in poorly secured enterprise systems, where adversaries may compromise some workers. In these settings, distributed computing becomes vulnerable to (i) privacy leakage and (ii) integrity attacks that degrade the accuracy of the final result. In practice, systems often possess prior information about the likelihood of worker unreliability. This information typically arises from historical behavior (e.g., fault rates, inconsistent outputs, reputation scores, operating systems deployed). Under such settings,  when the prior knowledge of these untrusted workers is available one can strategically assign encoded shares of the data to these workers to reduce both privacy leakage and accuracy degradation. To the best of our knowledge, such customized assignment has not been explored previously.

Towards this direction, we first propose a robust NS-LCC framework that accounts for untrusted  workers, including curious workers and Byzantine workers. In this context, we propose new encoding and reconstruction methods to provide privacy and robustness against a certain bounded number of curious and Byzantine workers. Further, in an untrusted environment, when the prior knowledge of these untrusted workers is available, we formulate optimization problems for the customized assignment of shares of the data to these untrusted workers, based on a derived bound on mutual information security (MIS) for privacy and localization error as a surrogate for accuracy. First, we solve the two optimization problems independently and observe that the solutions are not identical, revealing a clear privacy and robustness trade-off. We then address this trade-off by jointly solving the problem using a low-complexity algorithm.


\section{Robust Numerically Stable Lagrange Coded Computing}
\label{sec:basic NSLCC}

We consider the NS-LCC framework proposed in \cite{f1}, comprising a master node and $N$ worker nodes denoted by the set $\mathcal{P}=\{\mathcal{P}_1,\mathcal{P}_2,\ldots,\mathcal{P}_N\}$. Each worker communicates only with the master node through a dedicated link, while there is no communication among the workers. In addition, the dataset on which the computation is performed is held by the master and all workers have prior knowledge of the function $f$ to be evaluated on the dataset. In contrast to the vanilla NS-LCC framework, we consider a more realistic distributed computing framework in which a subset of workers may behave unreliably and include both $A$ number of Byzantine workers and $t$ number of curious workers. Specifically, $\mathcal{P}$ is partitioned into two disjoint subsets: reliable workers, denoted by $\mathcal{P}_{{rel}}$, and unreliable workers, denoted by $\mathcal{P}_{{unrel}}$, with $\rho = |\mathcal{P}_{{rel}}|$, $\nu= |\mathcal{P}_{{unrel}}|$, s.t. $\rho + \nu = N$. Reliable workers are assumed to behave honestly and return correct results, whereas unreliable workers may attempt to infer information about the dataset or return incorrect computations to the master. We assume that the master knows both $\mathcal{P}_{{rel}}$ and $\mathcal{P}_{{unrel}}$, as well as the identities of the workers in each set. Further, we assume that the identities of both curious and Byzantine workers are not known to the master node, however, they are assumed to belong to $\mathcal{P}_{{unrel}}$. In contrast, $s$ number of stragglers may arise arbitrarily from anywhere in $\mathcal{P}$. Let the dataset held by the master node be $\mathcal{X}=\{\mathbf{X}_1,\ldots,\mathbf{X}_k\}$ with $\mathbf{X}_{i} \in \mathbb{R}^{m\times n}$ for all $i \in [k]$. Under this setting, in the next subsection, we present a distributed approach for computing $f(\mathbf{X}_i)$ for all $i \in [k]$ in the presence of the above mentioned profiled workers.

\vspace{-0.1cm}
\subsection{Encoding }
\label{subsec:encoding}
We discuss the encoding method specifically highlighting how the dataset $\mathcal{X}$ is encoded into distributed shares for the workers. To encode the dataset $\mathcal{X}$, we use the same Lagrange polynomial based encoding method as in \cite{f1}, but modify it to ensure the privacy of $\mathcal{X}$ in the presence of $t$ curious workers. Specifically, master introduces $t$ random noise matrices into the encoding process in order to guarantee privacy against the dataset even if up to $t$ workers decides to collude. Accordingly, it encodes $\mathcal{W} = \{\mathbf{X}_{1}, \mathbf{X}_{2}, \ldots, \mathbf{X}_{k}, \mathbf{N}_{1}, \mathbf{N}_{2}, \ldots, \mathbf{N}_{t}\}$, instead of $\mathcal{X}$, where $\{\mathbf{N}_{1}, \mathbf{N}_{2}, \ldots, \mathbf{N}_{t}\}$ are noise matrices with $\mathbf{N}_{j} \in \mathbb{R}^{m \times n}$, such that $j \in [t]$. Each noise matrix is drawn from a zero--mean Gaussian distribution with standard deviation $\frac{\sigma_{n}}{\sqrt{t}}$, i.e., $\mathbf{N}_{j} \sim \mathcal{N}\!\left(0, \frac{\sigma_{n}^2}{{t}}\right)$, where $t$ is the maximum number of colluding workers. Let $\{\xi_{j}\}_{j=1}^{k+t}$ denote the $k+t$ encoding points, chosen as Chebyshev nodes of the first kind, given by $\xi_{j} = \cos\!\left(\frac{(2j-1)\pi}{2(k+t)}\right)$ for $j \in [k+t]$. Using these, the master constructs the encoding polynomial as

{\small
\begin{equation}
\label{eq:u(z) with privacy}
g(z)=\sum_{r=1}^{k} \mathbf{X}_r l_r(z) + \sum_{r=k+1}^{k+t} \mathbf{N}_{r-k}l_{r}(z),
\end{equation}}
where $l_r(.)$'s are Lagrange monomials defined as

{\small
\begin{equation}
\label{eq:lag monomial}
\ell_{j}(z)\,=\,\prod_{\substack{\ell=1 \\ \ell\neq j}}^{k+t}
\frac{z-\xi_{\ell}}{\xi_{j}-\xi_{\ell}}
\qquad j\in[k+t].
\end{equation}}

\subsection{Distribution of Shares among the Worker Nodes}
\label{subsec:distribution of shares}
Once the encoding polynomial is formed at the master node using $\mathcal{W}$ as shown in \eqref{eq:u(z) with privacy}, the shares of the encoded polynomial $g(z)$ are to be distributed to the worker nodes which comprises the evaluations of $g(z)$ over $\{\alpha_{i}\}_{i=1}^{N}$, where $\alpha_{i}\in \mathbb{R}$ indicate $N$ Chebyshev nodes of the first kind used for evaluation denoted as $\alpha_{i} = \cos\!\left(\frac{(2i-1)\pi}{2N}\right),i\in[N]$. More specifically, master node obtains the set of evaluations $\{\mathbf{Y}_{i}=g(\alpha_{i})~|~\alpha_i=\cos\!\left(\frac{(2i-1)\pi}{2N}\right)\}$. Further, to distribute the shares of the encoded dataset $g(z)$ among the worker nodes, we define a one-to-one mapping $\Phi:[N]\rightarrow [N]$, such that the $i$-th evaluation $\mathbf{Y}_{i}$ goes to the worker node $\mathcal{P}_{\phi(i)}$, for $i \in [N]$. 
\subsection{Computation at the Worker Nodes}
\label{subsec:computation at worker}
Once the worker $\mathcal{P}_{i}$ receives its evaluation $\mathbf{Y}_{i}$, it computes $f(\mathbf{Y}_{i})$, where $f$ is an arbitrary polynomial function such that $f:\mathbb{R}^{m\times n}\rightarrow\mathbb{R}^{u\times \bar{h}}$. When $A=0$, the master can use the set of computations returned by the non-straggling workers denoted by the set $\mathcal{G}$, i.e., $\{f(\mathbf{Y}_{i}),i\in \mathcal{G}\subseteq[N]\}$ for interpolation of $f(g(z))$. However, when $A>0$, some of them may behave as Byzantine and return corrupted results. In this context, let $i_{1}, i_{2}, \ldots, i_{A} \in \mathcal{G}$ denote the indices of the $A$ number of Byzantine workers, whose returned outputs are modeled as $f(\mathbf{Y}_{i_{a}}) + \mathbf{E}_{i_{a}}$, where $\mathbf{E}_{i_{a}} \in \mathbb{R}^{u\times \bar{h}}$ represents noise injected by the worker $\mathcal{P}_{i_{a}}$. In addition, all workers also introduce precision noise due to finite-precision arithmetic; therefore, the final computation results at the master are

{\footnotesize
\begin{equation}
\label{eq:rcv vector}
    \mathbf{R}_{i} = f(\mathbf{Y}_i) + \mathbf{E}_{i} + \mathbf{P}_{i} \in \mathbb{R}^{u \times \bar{h}},
    \quad i \in \{i_{1}, i_{2}, \ldots, i_{A}\},
\end{equation}
\begin{equation}
\label{eq:trust worthy compute}
    \mathbf{R}_{i} =  f(\mathbf{Y}_i) + \mathbf{P}_{i} \in \mathbb{R}^{u \times \bar{h}}
    \quad i \notin \{i_{1}, i_{2}, \ldots, i_{A}\} \cup [N\setminus\mathcal{G}],
\end{equation}
}

where $\mathbf{P}_{i}\in \mathbb{R}^{u \times \bar{h}}$ denotes the precision error due to floating--point operations. The terms $\{\mathbf{P}_{i}\}$ are independent across workers, and each entry is distributed as $\mathcal{N}(0,\sigma_{p}^{2})$. Note that \eqref{eq:trust worthy compute} corresponds to the computation results at the master in the absence of Byzantine workers.

\subsection{Function Reconstruction}

When $A=0$, for successful interpolation of $f(g(z))$, the master needs computation results from at least $(k+t-1)D_{f}+1$ workers. This follows from the fact that, the effective degree of $f(g(z))$ is $(k+t-1)D_{f}$, where $D_{f}$ indicates the degree of the target function $f(\cdot)$ and $k+t-1$ represents the degree of $g(z)$ indicated in \eqref{eq:u(z) with privacy}. Therefore, if $D_{eff}=(k+t-1)D_{f}$, denoting the effective degree of the polynomial $f(g(z))$, then the number of workers required for successful interpolation of $f(g(z))$ is $K = D_{eff} +1$. Once the master receives computation results from a set of non-straggling workers $\mathcal{G}$ with $|\mathcal{G}| = (k+t-1)D_{f}+1$, it interpolates the polynomial $f(g(z))$. Finally, to recover $f(\mathbf{X}_{i})$’s, the master computes $f(g(\xi_i))$ for $i\in[k]$, allowing the scheme to tolerate up to $s = N - K$ stragglers. Furthermore, when $A>0$, it is well known that the use of Chebyshev evaluation points enables the reconstruction stage to be modified via Discrete Cosine Transform (DCT)-based error detection and correction \cite{f3}, thereby providing robustness against $A$ Byzantine workers. Note that, in our setting, we assume that $0 < \rho < K$, since if $\rho \ge K$, any $K$ workers can be chosen for interpolation.

\vspace{-0.09cm}
\section{Privacy and Security Analyses of Robust NS-LCC}
\label{sec:robust NSLCC}
In this section, we study the privacy of $\mathcal{X}$ against $t$ curious workers and discuss the modification to reconstruction stage for securing the computations from $A$ Byzantine workers. Specifically, upon receiving $\mathbf{Y}_i$, by worker $\mathcal{P}_{i}$ we analyze the privacy of $\mathcal{X}$ against a subset of $t$ curious workers in terms of MIS leakage, as presented in the next subsection.

\subsection{Privacy against Curious Workers}
 For each matrix entry indexed by $g \in [m]$ and $h \in [n]$, let $X_{1}, X_{2}, \ldots, X_{k}$ and $N_{1}, N_{2}, \ldots, N_{t}$ denote the $(g,h)$-th elements of  $\mathbf{X}_{i}$ and $\mathbf{N}_{j}$, respectively, for $i \in [k]$ and $j \in [t]$. These scalars represent the data and noise symbols, where the data symbols are arbitrary random variables taking values in the interval $[-y,\, y]$, and each noise symbol satisfies $N_{j} \sim \mathcal{N}\!\left(0, \frac{\sigma^{2}_{n}}{t}\right)$. Furthermore, let $T=\{i_1,i_2,\ldots,i_t \}\subseteq [N]$ denote the set of colluding workers. The corresponding shares of the polynomial $g(z)$ received by these $t$ workers can therefore be expressed as
$Y_T = \mathbf{H}_T X + \mathbf{W}_TN$,
where
$X =[X_1,X_2,\ldots,X_k]^T$, $N=[N_1,N_2,\ldots,N_t]^T$ and $Y_T=[g(\alpha_{i_1}),g(\alpha_{i_2}),\ldots,g(\alpha_{i_t})]^T$ represents the corresponding entry of $\mathbf{Y}_{i}$, for $i\in[N]$. The matrices $\mathbf{H}_T \in \mathbb{R}^{t\times k}$ and $\mathbf{W}_T \in \mathbb{R}^{t\times t}$, which contain the evaluations of the Lagrange basis functions defined in \eqref{eq:lag monomial}, can be expressed as
\begin{equation}
    \label{eq:HT_WT_scaled}
    \begingroup
    \setlength{\arraycolsep}{2pt}
    \resizebox{\columnwidth}{!}{$
    \mathbf{H}_T =
    \begin{bmatrix}
    l_1(\alpha_{i_1}) & \cdots & l_k(\alpha_{i_1}) \\
    \vdots & \ddots & \vdots \\
    l_1(\alpha_{i_t}) & \cdots & l_k(\alpha_{i_t})
    \end{bmatrix}_{t\times k}
    \mathbf{W}_T =
    \begin{bmatrix}
    l_{k+1}(\alpha_{i_1}) & \cdots & l_{k+t}(\alpha_{i_1}) \\
    \vdots & \ddots & \vdots \\
    l_{k+1}(\alpha_{i_t}) & \cdots & l_{k+t}(\alpha_{i_t})
    \end{bmatrix}_{t\times t}
    $}
    \endgroup
\end{equation}

where $\{\alpha_{i_{1}}, \alpha_{i_{2}}, \ldots, \alpha_{i_{t}}\}$ denote the subset of Chebyshev nodes of the first kind corresponding to the colluding set $T$, where $\alpha_i = \cos\!\big(\frac{(2i-1)\pi}{2N}\big)$. Further, we follow the MIS analysis in \cite{b6}, originally derived for roots of unity, and generalize it to Chebyshev nodes of the first kind. Accordingly, the amount of information revealed to the colluding set $T$ is measured by the MIS metric, denoted by $\eta_c$, and defined as

\begin{equation}
\label{eq:first MIS}
    \eta_c\triangleq 
    \max_{T\subseteq N,\; |T| \le t} \;
    \max_{P_\mathcal{X}:\, |X_i| \le y}
    I(Y_T; \mathcal{X}),
\end{equation}
where $P_{\mathcal{X}}$ is the probability density function of $\mathcal{X}$, and the maximization is taken over all subsets $T$ of size at most $t$. Since $|X_j| \le y$, we have $\mathbb{E}[X_j^2] \le y^2$. Thus, \eqref{eq:first MIS} can equivalently be written as
\begin{equation}
\label{eq:MIS bound}
    \eta_c=\max_{T \subseteq N,\; |T| \le t} \;
    \max_{P_{\mathcal{X}}:\, \mathbb{E}[X_j^2] \le y^2}
    I(\mathcal{X}; Y_T).
\end{equation}

\begin{proposition}
\label{prop:MIS}
    For a given $N$, $k$, $t$, $y$, $\sigma_{n}$, and when $y=O(\sigma_n)$ the MIS metric defined in \eqref{eq:MIS bound} can be upper bounded as 
    \begin{equation}
\label{eq:MIS second bound}
\eta_c\leq\frac{1}{\ln 2}\,\max_{T\subseteq N}\,\mathrm{tr}\!\left(\mathbf{\tilde{\Sigma}}_T^{-1}\mathbf{\Sigma_T}\right)\frac{y^2 t}{\sigma_n^2}+ o\!\left(\frac{y^2}{\sigma_n^2}\right),
\end{equation}
where $\mathbf{\Sigma}_T=\mathbf{H}_T \mathbf{H}_T^T$ and $\mathbf{\tilde{\Sigma}}_T = \mathbf{W}_T \mathbf{W}_T^T$ and $\mathbf{H}_T$ and $\mathbf{W}_T$ is defined in \eqref{eq:HT_WT_scaled}.
\end{proposition}
Note that the bound in \eqref{eq:MIS second bound} decreases as $\sigma_n$ increases. Furthermore, it is dominated by the trace term $\mathrm{tr}\!\left(\tilde{\mathbf{\Sigma}}_T^{-1}\mathbf{\Sigma}_T\right)$, which is determined by the evaluation indices assigned to the $t$ curious workers. Hence, without profiled workers, the minimum leakage is determined by the subset that maximizes the trace over all $\binom{N}{t}$ combinations. However, with profiled workers, the $t$ colluding workers are constrained to lie within the set of unreliable workers. This allows the master to strategically assign their evaluation indices in way that minimize the leakage $\eta_c$.


\subsection{Robustness against Byzantine Workers}
\label{subsec:secure NSLCC}
After receiving its share $\mathbf{Y}_{i}$ from the master node, each worker $\mathcal{P}_{i}$ computes $\mathbf{R}_{i} = f(\mathbf{Y}_{i}) \in \mathbb{R}^{u \times \bar{h}}$ and returns the result to the master. In the presence of stragglers and unreliable workers where some of them may be Byzantine, master collects the computations results returned by the non-straggling workers, whose indices are given by $\mathcal{G} \!=\!\{i_{1}, i_{2}, \ldots, i_{N-s}\}$ as indicated in \eqref{eq:rcv vector}, s.t. $|\mathcal{G}|=N-s$. Let the corresponding results received from these workers are denoted by $\mathcal{R}=\{\mathbf{R}_{i_{1}},\, \mathbf{R}_{i_{2}},\, \ldots,\, \mathbf{R}_{i_{N-s}}\}$, where $\mathbf{R}_{i}\in\mathbb{R}^{u\times \bar{h}}$ for $i\!\in\! \mathcal{G}$. For each pair of indices $g \in \{1,2,\ldots,u\}$ and $h \!\!\in\!\! \{1,2,\ldots,\bar{h}\}$, we define the vector $\mathbf{r}_{g,h}\! \!=\! \![\,\mathbf{R}_{i_{1}}(g,h)\;\mathbf{R}_{i_{2}}(g,h)\;\ldots\;\mathbf{R}_{i_{N-s}}(g,h)\,]$, constructed from the $(g,h)$-th entry of the outputs from all workers in $\mathcal{G}$. Thus, $\mathbf{r}_{g,h}$ represents a noisy vector of length $N - s$. In total, the master receives $L = u \times \bar{h}$ such noisy vectors, each of length $N-s$. To detect and correct the errors introduced by the Byzantine workers within these $L$ received vectors, we present the following propositions.
\begin{proposition}
\label{prop:dct proof}
The vector $\mathbf{r}_{g,h} \in \mathbb{R}^{N-s}$, $\forall g, h$ can be represented as a noisy codeword of a $K$-dimensional Discrete Cosine Transform (DCT) code of block length $N-s$ where $K=(k+t-1)D_{f}+1$.
\end{proposition}

\begin{proof}
Proof is along the lines of \cite[Proposition~2]{f3}.
\end{proof}

\begin{proposition}
\label{prop:err corr}
For given $N$, $k$, $t$, $s$, $K$ and $A$, the errors introduced by $A$ Byzantine workers can be detected and nullified, as long as $A \le \left\lfloor \frac{(N-s)-K}{2} \right\rfloor$ and $\sigma_p^2 = 0$.
\end{proposition}

Due to Proposition \ref{prop:dct proof} and Proposition \ref{prop:err corr}, the structure of the received vector $\mathbf{r}_{g,h}$ enables the use of a DCT-based error detection and correction algorithm \cite{b11} to identify and correct errors introduced by Byzantine workers. From \eqref{eq:rcv vector}, we have $\mathbf{r}_{g,h} = f(\mathbf{y}_{g,h}) + \mathbf{p}_{g,h} + \mathbf{e}_{g,h}$, where $\mathbf{p}_{g,h}$ and $\mathbf{e}_{g,h}$ denote the precision noise and the Byzantine error vectors, respectively. Using the syndrome obtained by $\mathbf{r}_{g,h}$, we first estimate the number of errors $A \leq \left\lfloor\frac{(N-s)-K}{2}\right\rfloor$ and then compute the coefficients of estimated error-locator polynomial $\bar{\Lambda}_{g,h}(x) = \Lambda_{g,h}(x) + p_{g,h}(x)$. Ideally, when $\sigma_p^2=0$, the roots of $\Lambda_{g h}(x)$ directly provides the actual error locations. However, under finite-precision arithmetic, we must evaluate $\|\bar{\Lambda}_{g h}(X_{r})\|^{2}$ for each $r\in \mathcal{G}$, where
$X_{r} = \cos\!\left(\frac{(2r-1)\pi}{2N}\right)$ and sort these values in ascending order as $\|\bar{\Lambda}_{g h}(X_{{i}_{1}})\|^{2} \leq \|\bar{\Lambda}_{g h}(X_{{i}_{2}})\|^{2} \leq \cdots \leq \|\bar{\Lambda}_{g h}(X_{{i}_{N-s}})\|^{2},$
where $\{{i}_{1}, {i}_{2}, \ldots, {i}_{N-s}\}$ denotes the ordered set of indices. The $A$ smallest values and their corresponding indices $\{{i}_{1}, {i}_{2}, \ldots, {i}_{A}\}$ are then identified as the detected error locations. Further, let the set of indices corresponds to the unreliable workers be denoted by the set $\mathcal{V}$. Since we assume that all $A$ Byzantine workers belong to the set of unreliable workers $\mathcal{V}$, we only need to evaluate $\|\bar{\Lambda}_{g h}(X_{r})\|^{2}$ for indices $r \in \mathcal{V}$. In particular, instead of evaluating $\bar{\Lambda}_{g,h}(x)$ at all indices corresponding to the set $\mathcal{G}$, we may restrict the evaluation to the subset of indices corresponding to the set $\mathcal{V}$. Although $\bar{\Lambda}_{g,h}(x)$ is evaluated over the indices in the set $\mathcal{G}$, in the following analysis we assume $s = 0$, i.e., $\mathcal{G} = [N]$. The same analysis can be extended to the case when $s > 0$. In this context, let the true set of error locations be $\mathcal{A} = \{i_{1}, i_{2}, \ldots, i_{A}\}$, and the set of recovered locations be $\hat{\mathcal{A}} = \{\hat{i}_{1}, \hat{i}_{2}, \ldots, \hat{i}_{A}\}$. Under the assumption that all Byzantine workers belong to the set $\mathcal{V}=\{k_{u_1}, k_{u_2}, \ldots, k_{u_\nu}\} $, the error probability associated with mislocalizing at least one of the $A$ error positions is $P_{{error}} = \Pr(\hat{\mathcal{A}} \neq \mathcal{A})$, which can be upper-bounded as

{\footnotesize
\begin{equation}
P_{{error}} \le \sum_{i=1}^{\nu - A} \sum_{a=1}^{A}
\Pr\!\left(\|\bar{\Lambda}_{g h}(X_{k_{u_{i}}})\|^{2} \le
\|\bar{\Lambda}_{g h}(X_{i_{a}})\|^{2}\right).
\end{equation}}
\vspace{-0.08cm}
Since the Byzantine workers are chosen as a subset of $\mathcal{V}$ and each subset of size $A$ is chosen uniformly at random over $\mathcal{V}$, the average probability of localization error is expressed as an expectation over all possible choices of $\mathcal{A} \subseteq \mathcal{V}$. In particular, the average probability of incorrect error localization is
\vspace{-0.1cm}

{
\scriptsize
\begin{IEEEeqnarray}{rcl}
\label{eq:avg error prob}
\bar{P}_{{error}}
= \mathbb{E}_{\mathcal{A}\subseteq \mathcal{V}}\bigg[\sum_{i=1}^{\nu-A}\sum_{a=1}^{A}\Pr\!\left(
\|\bar{\Lambda}_{g h}(X_{k_{u_{i}}})\|^{2}\le\|\bar{\Lambda}_{g h}(X_{i_{a}})\|^{2}\right)
\bigg].
\end{IEEEeqnarray}
}

Since $\bar{P}_{{error}}$, defined in \eqref{eq:avg error prob}, depends on the evaluation of $\bar{\Lambda}_{g,h}(x)$ at the indices assigned to the unreliable workers in $\mathcal{V}$ (which include the Byzantine workers), we can assign the evaluation indices to the set of workers in $\mathcal{V}$ in a manner that minimizes $\bar{P}_{{error}}$.

\section{Privacy and Localization Error based Assignment}
From a system perspective, one may use privacy or accuracy as evaluation metrics. Here, we highlight the trade-off between them by jointly considering privacy maximization and accuracy improvement. Since a closed-form expression for the accuracy metric is not available as a function of system parameters, we use localization error as a surrogate metric for accuracy. Accordingly, we propose strategies to assign evaluation indices to unreliable workers by minimizing the bound on MIS and the average probability of localization error.

\subsection{Assignment based on MIS }
Note that the bound in \eqref{eq:MIS second bound} in Proposition \ref{prop:MIS} involves computing the maximum trace term by considering all possible $\binom{N}{t}$ choices of $t$ colluding workers. However, under our assumption that the $t$ curious workers must come from the set of unreliable workers $\mathcal{P}_{unrel}$, we can significantly reduce this search space. Let the indices of the $\nu$ unreliable workers be denoted by $\mathcal{V}=\{k_{u_1}, k_{u_2}, \ldots, k_{u_\nu}\} \subseteq [N]$,
where $|\mathcal{V}| = \nu$. Since the $t$ colluding workers must lie within this set, the colluding set $T = \{i_1, i_2, \ldots, i_t\}$ is constrained to satisfy $T \subseteq \mathcal{V}$. Therefore, instead of maximizing over all $\binom{N}{t}$ subsets of size $t$, we only need to consider subsets of size $t$ drawn from $\mathcal{V}$. Therefore, the MIS leakage bound in \eqref{eq:MIS second bound} is 

\begin{equation}
\label{eq:MIS_unreliable}
\eta_c \leq \frac{1}{\ln 2}\, \max_{T \subseteq \mathcal{V}} 
\mathrm{tr}\!\left(\mathbf{\tilde{\Sigma}}_{T}^{-1}\mathbf{\Sigma}_{T}\right)
\frac{y^{2} t}{\sigma_n^{2}}
+ o\!\left(\frac{y^{2}}{\sigma_n^{2}}\right).
\end{equation}

Furthermore, the above bound depends on the magnitude of the trace term, which in turn is determined by the matrices $\mathbf{H}_{T}$ and $\mathbf{W}_{T}$ defined in \eqref{eq:HT_WT_scaled}, respectively. These matrices are constructed from evaluations of the Lagrange basis functions at the points $\alpha_{i}$. Therefore, since the master node knows the identities of the unreliable workers, the subset of evaluation points $\alpha_{i}$ for $i\in[\nu]$ assigned to these workers for evaluation of $g(z)$ must therefore be chosen carefully, as any group of $t$ colluding workers will lie within this set and their associated evaluation points directly influence the magnitude of $\mathrm{tr}\!\left(\mathbf{\tilde{\Sigma}}_{T}^{-1}\mathbf{\Sigma}_{T}\right)$ and thereby MIS leakage. 

We highlight that using the identity mapping i.e., $\phi(i)=i$ for assigning evaluation points of $g(z)$ may not be an optimal choice. Instead, a customized mapping can be used for assigning evaluations specifically to the unreliable worker nodes such that it minimizes MIS leakage bound $\eta_c$ defined in \eqref{eq:MIS_unreliable}. More precisely, let $\mathcal{Q} \subseteq [N]$ denote the set of $\nu$ distinct indices corresponding to the Chebyshev nodes selected for evaluating $g(z)$. Based on this selection, we use a mapping $\phi(\cdot)$ as described in Section~\ref{subsec:distribution of shares}, such that for every index $j \in \mathcal{V}$, there exists a unique $i \in \mathcal{Q}$ satisfying $\phi(i) = j$. This ensures that each unreliable worker $\mathcal{P}_{j}$ receives the evaluation $\mathbf{Y}_{i}$ associated with its assigned index $i\in\mathcal{Q}$. The remaining evaluations $\{\mathbf{Y}_{i} : i \notin \mathcal{Q}\}$ may be assigned arbitrarily among the reliable workers in $\mathcal{P}_{{rel}}$. Since the master node has $\binom{N}{\nu}$ possible choices for selecting $\mathcal{Q}$, the goal is to identify the optimal subset $\mathcal{Q}^{*} \subseteq [N]$ that minimizes the bound in \eqref{eq:MIS_unreliable}. More precisely, we propose to solve the following problem.

\begin{mdframed}
\footnotesize
\begin{small}
\begin{problem}
\label{opt1}
For the given parameters $N$, $t$, $\nu$, $y$, and $\sigma_{n}$, solve
\begin{IEEEeqnarray*}{rcl}
\mathcal{Q^*}=\arg\min_{\substack{\mathcal{Q} \subseteq [N] \\ \text{s.t. } |\mathcal{Q}|=\nu}}
\Bigg\{
\frac{1}{\ln 2}\,
\max_{\mathcal{T} \subseteq \mathcal{Q}}\,
\mathrm{tr}\!\left(\mathbf{\tilde{\Sigma}}_{T}^{-1}\mathbf{\Sigma}_{T}\right)
\frac{y^{2} t}{\sigma_{n}^{2}}
+ o\!\left(\frac{y^{2}}{\sigma_{n}^{2}}\right)
\Bigg\}.
\end{IEEEeqnarray*}
\end{problem}
\end{small}
\end{mdframed}
Here, $\mathcal{T}=\{i_{1}, i_{2}, \ldots, i_{t}\} \subseteq \mathcal{Q}$ denotes the set of indices corresponding to the colluding workers used for evaluating $g(z)$. For a given choice of system parameters, we solve Problem \ref{opt1} and present the resulting set $\mathcal{Q}^{*}$ in Table \ref{tab: z* and q*}. Note that $\mathcal{Q}^{*}$ consists of the $\nu$ indices that maximize the trace over all $\binom{N}{\nu}\binom{\nu}{t}$ combinations, and these indices typically correspond to the extreme points of the evaluation points $\{\alpha_{i}\}_{i=1}^{N}$. 

\vspace{-0.19cm}
\subsection{Assignment based on Localization Error}
Note that the bound of $\bar{P}_{error}$ in \eqref{eq:avg error prob}, depends on the evaluation of noisy error-locator polynomial $\bar{\Lambda}_{g,h}(x)$ at the subset of Chebyshev nodes assigned to the unreliable workers in $\mathcal{V}$. Since  $\bar{P}_{error}$ depends on these evaluation points, the choice of evaluation indices assigned to the set $\mathcal{V}$ must be chosen carefully. Specifically, for a given $\mathcal{V}$, we determine a corresponding set of evaluation indices $\mathcal{Q}$ used to assign the evaluations of $g(z)$ to the workers in set $\mathcal{V}$ with the aim of minimizing  $\bar{P}_{error}$. In this context, let $\mathcal{Q} = \{l_{u_1}, l_{u_2}, \ldots, l_{u_\nu}\}$ denote the set of evaluation indices assigned to the unreliable worker set $\mathcal{V}$, such that for every index $j \in \mathcal{V}$, there exists a unique $i \in \mathcal{Q}$ satisfying $\phi(i) = j$. This ensures that each unreliable worker $\mathcal{P}_{j}$ receives the evaluation $\mathbf{Y}_{i}$ associated with its assigned evaluation index $i \in \mathcal{Q}$. Since there are $\binom{N}{\nu}$ possible choices for $\mathcal{Q}$, our objective is to identify the optimal subset $\mathcal{U}^{*} \subseteq [N]$ that minimizes the average probability of localization error $\bar{P}_{{error}}$ defined in \eqref{eq:avg error prob}. Formally, we propose

{\scriptsize
\begin{equation}
\label{eq:adv_opt}
\mathcal{\mathcal{U}}^* =
\arg\min_{\substack{\mathcal{Q}\subseteq [N] \\ |\mathcal{Q}|=\nu}}
\mathbb{E}_{\mathcal{A}\subseteq \mathcal{Q}}
\!\left[
\sum_{i=1}^{\nu-A}\sum_{a=1}^{A}
\Pr\!\Big(
\|\bar{\Lambda}_{gh}(X_{l_{u_i}})\|^2 \le
\|\bar{\Lambda}_{gh}(X_{i_a})\|^2
\Big)
\right].
\end{equation}}

Furthermore, we derive a lower bound on the pairwise error probability (PEP), given by, $\Pr\!\big(\|\bar{\Lambda}_{gh}(X_{l_{u_i}})\|^2 \leq\|\bar{\Lambda}_{gh}(X_{i_a})\|^2\big)$, for a given pair of indices $l_{u_i} \in \mathcal{Q}$ and $i_a \in \mathcal{A}$. Using this bound,  we first define the following surrogate function for $\bar{P}_{error}$,
\begin{equation}
\label{eq:avg error expression}
\tilde{P}_{\mathrm{error}}(\mathcal{Q})
\triangleq
\mathbb{E}_{\mathcal{A}\subseteq \mathcal{Q}}
\Bigg[
\max_{i\in \{1,2,\ldots,\nu-A\}}
e^{-\frac{\zeta f(\mathcal{A},i)\delta_{\max}}{8\sigma_p^2}}
\Bigg].
\end{equation}
Using \eqref{eq:avg error expression}, we rewrite \eqref{eq:adv_opt} and finally propose to solve

\begin{mdframed}

\begin{problem}
\label{opt2}
For given $N$, $k$, $t$, $A$, $\nu$, $\sigma_p^2$, and $\zeta$, solve
\begin{IEEEeqnarray*}{rcl}
\mathcal{Z}^*
=\arg\min_{\substack{\mathcal{Q} \subseteq [N] \\ |\mathcal{Q}|=\nu}}
\mathbb{E}_{\mathcal{A}\subseteq \mathcal{Q}}
\Bigg[
\max_{i\in \{1,2,\ldots, \nu-A\}}
e^{-\frac{\zeta f(\mathcal{A},i)\delta_{\max}}{8\sigma_p^2}}
\Bigg].
\end{IEEEeqnarray*}
\end{problem}
\end{mdframed}

Here, $f(\mathcal{A},i)\!=\!\left|\prod_{a=1}^{A} \left(\cos \theta_{l_{u_i}}-\cos \theta_{i_a} \right) \right|^{2}$, where $\theta_{l_{u_i}}\!\!=\!\!\frac{(2l_{u_i}-1)\pi}{2N}$ and $\theta_{i_a}\!\! =\! \frac{(2i_a-1)\pi}{2N}$, for $l_{u_i}\!\!\in\!\mathcal{Q}$ and $i_a\in\mathcal{A}$, $\zeta$ is a constant and $\frac{\beta}{1+\beta} \le \delta_{\max}$, with $\beta = \frac{4}{\zeta \sum_{k=1}^{A} (\cos^{k}\theta_{l_{u_i}} -\cos^{k}\theta_{i_a})^{2}}$. We solve Problem \ref{opt2} and present the resulting set $\mathcal{Z}^{*}$ in Table \ref{tab: z* and q*} for a given set of parameters under varying $\sigma_{p}^{2}$. Note that, unlike $\mathcal{Q}^{*}$, the set $\mathcal{Z}^{*}$ consists of well separated evaluation indices. Further, for a given $N, \nu, t, A, y,$ and $\sigma_n$, the solution to Problem \ref{opt1} depends only on the parameters of encoding stage, whereas the solution to Problem \ref{opt2} additionally depends on the precision noise introduced at the workers during computation due to finite precision. Therefore, we evaluate Problem \ref{opt2} by varying $\sigma_p^2$, while solution to the Problem \ref{opt1} remains same. We present these results in Table \ref{tab: z* and q*} for different system parameters. As seen from Table  \ref{tab: z* and q*}, solutions to the Problems \ref{opt1} and \ref{opt2} are not identical. In particular, for a given $\sigma_{p}^{2}$, Problems \ref{opt1} and \ref{opt2} have no common solution, i.e., $\mathcal{Q}^*$ and $\mathcal{Z}^*$ are not identical.

\begin{table*}[ht]
    \centering
    \caption{$\mathcal{Q}^{*}$ and $\mathcal{Z}^{*}$ under varying $\sigma_{p}^2$ with parameters $N=21$, $D_{f}=2$, $\nu=12$, $k=3$, $t=A=3$, $y=10^{10}$, $\sigma_n=10^{23}$, $\zeta=100$.}
    \label{tab: z* and q*}
    \renewcommand{\arraystretch}{1}
    \footnotesize
    \resizebox{0.8\textwidth}{!}{%
    \begin{tabular}{|c|p{0.24\textwidth}|p{0.24\textwidth}|p{0.24\textwidth}|}
    \hline
    \diagbox[width=3.5em,height=2.2em]{set}{$\sigma_{p}^2$} 
     & $1\times10^{-2}$ & $1\times10^{-3}$ & $1\times10^{-4}$
    \tabularnewline
    \hline

$\mathcal{Z}^{*}$ &
1 3 5 7 8 10 11 12 14 15 17 20 &
1 3 5 6 7 9 10 12 13 14 17 20 &
1 3 5 7 8 9 10 11 13 14 17 19
\\ \hline

$\mathcal{Q}^{*}$ &
4 11 12 13 14 15 16 17 18 19 20 21 &
4 11 12 13 14 15 16 17 18 19 20 21 &
4 11 12 13 14 15 16 17 18 19 20 21
\\ \hline
\end{tabular}
}
\end{table*}


\newcolumntype{P}[1]{>{\rule{0pt}{1.5ex}}p{#1}}
\begin{table*}[ht]
\centering
\caption{$\mathcal{S}_g^{*}$ under different $\sigma_{p}^2$ after solving Problem \ref{opt:joint opt}. The parameters used for the experiments are same as used in Table \ref{tab: z* and q*}.}
\label{tab:variance_sets}
\renewcommand{\arraystretch}{1}
\footnotesize
\resizebox{0.8\textwidth}{!}{%
\begin{tabular}{|c|p{0.24\textwidth}|p{0.24\textwidth}|p{0.24\textwidth}|}
\hline
\diagbox[width=3.5em,height=2.2em]{$w$}{$\sigma_{p}^2$} 
 & \centering $1\times10^{-2}$ & \centering $1\times10^{-3}$ & \centering $1\times10^{-4}$
\tabularnewline
\hline

0 &
1 3 5 7 8 10 11 12 14 15 17 20 &
1 3 5 6 7 9 10 12 13 14 17 20 &
1 3 5 7 8 9 10 11 13 14 17 19
\\ \hline

0.2 &
1 3 5 7 8 10 11 12 14 15 17 20 &
1 4 5 7 8 9 10 11 13 14 17 20 &
1 4 6 8 9 10 11 12 14 15 18 20
\\ \hline

0.4 &
1 3 5 7 8 10 11 12 14 15 17 20 &
1 4 5 7 8 9 10 11 13 14 17 20 &
1 4 6 8 9 10 11 12 14 15 18 20
\\ \hline

0.6 &
1 3 5 7 8 10 11 12 14 15 17 20 &
1 4 5 7 8 9 10 11 13 14 17 20 &
3 4 6 8 9 10 11 12 14 15 18 20
\\ \hline

0.8 &
1 3 5 7 8 10 11 12 14 15 17 20 &
1 4 5 7 8 9 10 11 13 14 17 20 &
3 4 6 8 9 10 11 12 14 15 18 20
\\ \hline

1.0 &
4 11 12 13 14 15 16 17 18 19 20 21 &
4 11 12 13 14 15 16 17 18 19 20 21 &
4 11 12 13 14 15 16 17 18 19 20 21
\\ \hline

\end{tabular}}
\end{table*}



\begin{table*}[t]
\centering
\caption{$\mathcal{S}^*$ and $\mathcal{S}^*_{g}$ under varying $\sigma_p^2$, with $N = 13$, $D_{f}=2$, $\nu = 8$, $k = 3$, $t = A = 2$, $y = 10^{10}$, $\sigma_{n} = 10^{23}$, and $\zeta = 100$.}
\label{tab:bf_greedy_variance}
\renewcommand{\arraystretch}{1}
\footnotesize
\begin{tabular}{|c|p{0.20\textwidth}|p{0.20\textwidth}|p{0.20\textwidth}|p{0.20\textwidth}|}
\hline
\multirow{2}{*}{$w$} 
& \multicolumn{2}{c|}{$\mathcal{S}^*$ } 
& \multicolumn{2}{c|}{$\mathcal{S}^*_{g}$ } \\
\cline{2-5}
& ${1\times10^{-2}}$ & ${1\times10^{-4}}$
& ${1\times10^{-2}}$ & ${1\times10^{-4}}$ \\
\hline

0 &
1 3 5 6 8 9 11 14 &
1 3 5 6 8 9 11 14 &
1 3 5 6 7 9 11 14 &
1 2 5 6 7 9 10 14 \\
\hline

0.4 &
1 3 5 6 7 9 10 14 &
2 5 6 7 9 10 12 14  &
1 3 5 6 7 9 11 14 &
2 3 5 6 7 10 11 14 \\
\hline

0.8 &
1 3 5 6 7 9 10 14 &
2 5 6 7 9 10 12 14 &
1 3 5 6 7 9 11 14 &
2 3 5 6 7 10 11 14 \\
\hline

1.0 &
2 5 10 11 12 13 14 15 &
2 5 10 11 12 13 14 15 &
2 5 10 11 12 13 14 15 &
2 5 10 11 12 13 14 15 \\
\hline

\end{tabular}
\end{table*}

{\tiny
\begin{algorithm}[t]
\caption{Greedy algorithm for computing $\mathcal{S}_{g}^*$}
\label{algo:1}
\KwIn{$N$, $\nu$, $t$, $A$, $y$, $\sigma_n$, $\sigma_p^2$, and $w \in [0,1]$}
\KwOut{$\mathcal{S}^{*}_{g}$}

$m \leftarrow \max(A,t)$\;
$\mathcal{C}_{\text{initial}} \leftarrow \{\mathcal{Q} \subseteq [N] : |\mathcal{Q}| = m\}$\;

\textbf{Initial Selection:}
\ForEach{$\mathcal{Q} \in \mathcal{C}_{\text{initial}}$}{
    Compute $\eta_c(\mathcal{Q})$ using \eqref{eq:MIS_unreliable} with $\mathcal{V}=\mathcal{Q}$\;
    Compute $\bar{P}_{{error}}(\mathcal{Q})$ using \eqref{eq:avg error expression} with $\mathcal{V}=\mathcal{Q}$\;
    $J(\mathcal{Q}) \leftarrow w \cdot \eta_c(\mathcal{Q}) + (1-w)\cdot \bar{P}_{{error}}(\mathcal{Q})$\;
}
$\mathcal{V} \leftarrow \arg\min\limits_{\mathcal{Q} \in \mathcal{C}_{\text{initial}}} J(\mathcal{Q})$\;
$\mathcal{R} \leftarrow [N] \setminus \mathcal{V}$\;

\textbf{\textit{Greedy} Expansion:}
\While{$|\mathcal{V}| < \nu$}{
    \ForEach{$p \in \mathcal{R}$}{
        $\mathcal{Q} \leftarrow \mathcal{V} \cup \{p\}$\;
        Compute $\eta_c(\mathcal{Q})$ using \eqref{eq:MIS_unreliable} with $\mathcal{V}=\mathcal{Q}$\;
        Compute $\bar{P}_{{error}}(\mathcal{Q})$ using \eqref{eq:avg error expression} with $\mathcal{V}=\mathcal{Q}$\;
        $J(p) \leftarrow w \cdot \eta_c(\mathcal{Q}) + (1-w)\cdot \bar{P}_{{error}}(\mathcal{Q})$\;
    }
    $p^* \leftarrow \arg\min\limits_{p \in \mathcal{R}} J(p)$\;
    $\mathcal{V} \leftarrow \mathcal{V} \cup \{p^*\}$\;
    $\mathcal{R} \leftarrow \mathcal{R} \setminus \{p^*\}$\;
}
\Return{$\mathcal{S}_g^* = \mathcal{V}$}
\end{algorithm}}

\begin{figure}[ht]
    \centering
    \includegraphics[width=0.88\columnwidth]{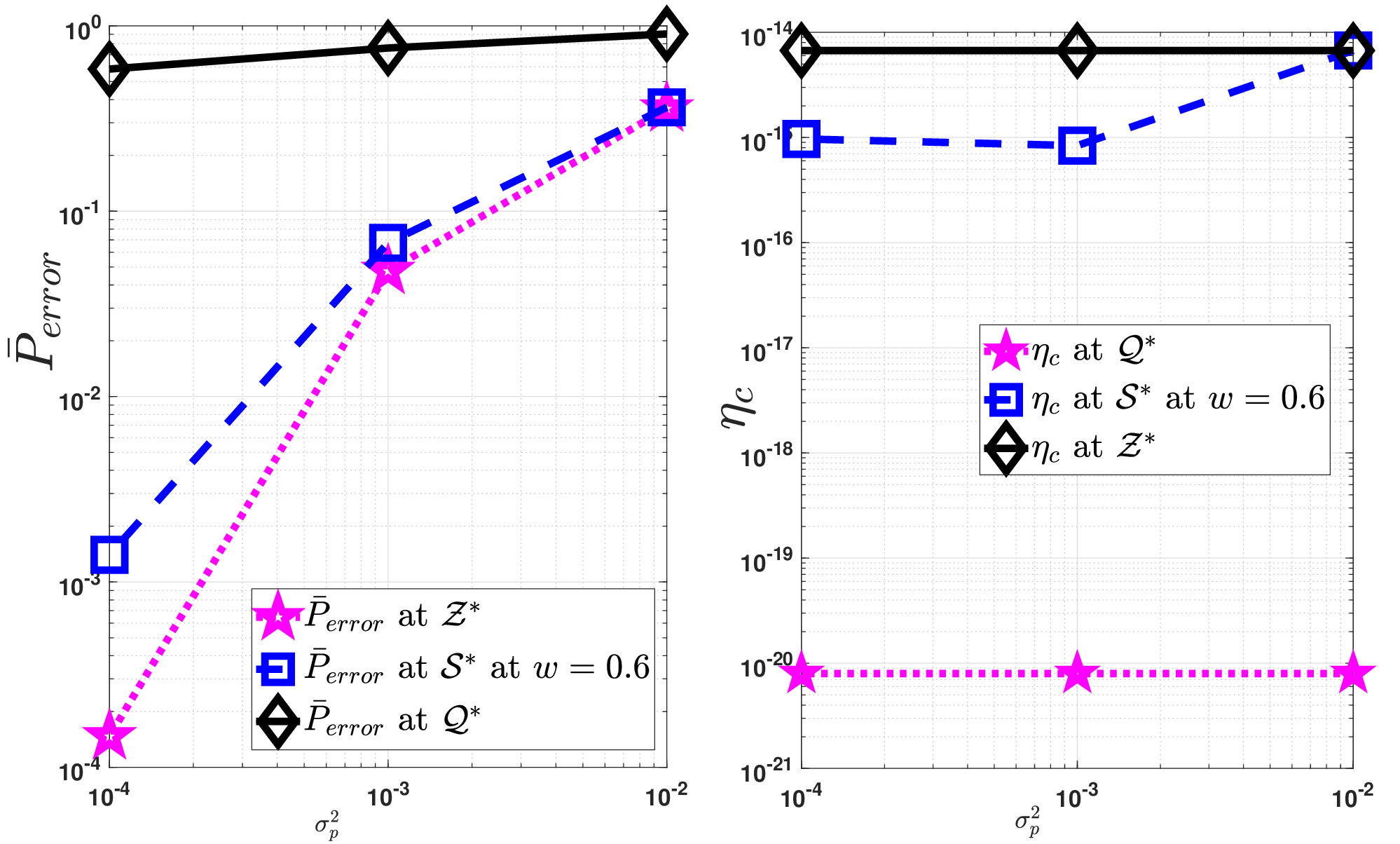}
    \caption{Results at $\mathcal{S}^*_g$ for different values of $w$ and $\sigma_p^2$ with parameters $N=21$, $\nu=12$, $k=3$, $t=A=3$, $y=10^{10}$, $\sigma_n=10^{23}$, $\zeta=100$.}
    \label{fig:four_in_row}
\end{figure}

To further investigate this behavior, we compute $\eta_{c}$ obtained when using $\mathcal{Z}^{*}$, and $\bar{P}_{{error}}$ obtained when using $\mathcal{Q}^{*}$. These results are illustrated in Fig.\ref{fig:four_in_row}. As shown in Fig. \ref{fig:four_in_row}, using $\mathcal{Q}^{*}$ to assign evaluation indices to the unreliable worker set $\mathcal{V}$ (with the goal of minimizing $\bar{P}_{{error}}$) can lead to a substantial increase in $\bar{P}_{{error}}$. Conversely, using $\mathcal{Z}^{*}$ to assign evaluation indices (with the goal of minimizing $\eta_{c}$) can significantly increase $\eta_{c}$. This trade-off leads us to jointly optimize both metrics. Therefore, in the next subsection, we present a strategy for obtaining these evaluation indices that simultaneously minimizes both $\eta_{c}$ and $\bar{P}_{{error}}$.

\subsection{Assignment based on Joint MIS and  Localization Error}
Provided the solutions of Problem \ref{opt1} and Problem \ref{opt2} are, in general, not identical, a fundamental trade-off arises between minimizing $\eta_c$ and $\bar{P}_{{error}}$. Therefore, in this section, we seek to identify the optimal set of indices $\mathcal{S}^{*}$ for assigning evaluations to the set of workers in $\mathcal{V}$, that jointly minimizes both $\eta_{c}$ and $\bar{P}_{{error}}$. Formally, we propose to solve

\begin{mdframed}
\footnotesize
\begin{problem}
\label{opt:joint opt}
Given $N$, $\nu$, $t$, $A$, $y$, $\sigma_n$, $\sigma_p^2$, and $w \in [0,1]$, solve
\begin{equation}
\mathcal{S}^{*}
= \arg\min_{\substack{\mathcal{Q} \subseteq [N], |\mathcal{Q}| = \nu}}
\Big\{
w\,\eta_{c}(\mathcal{Q}) + (1-w)\,\bar{P}_{{error}}(\mathcal{Q})
\Big\},
\end{equation}

where $\eta_{c}(\mathcal{Q})$ is the MIS leakage defined in \eqref{eq:MIS_unreliable} and $\bar{P}_{{error}}(\mathcal{Q})$ is the average probability of localization error corresponding to the subset $\mathcal{Q}$ defined in \eqref{eq:avg error expression}.
\end{problem}
\end{mdframed}


Since the master node knows the set of indices correspond to the unreliable workers i.e., $\mathcal{V}$, it can distribute the shares of $g(z)$ using the evaluation indices $\mathcal{S}^{*}$ by solving Problem \ref{opt:joint opt} such that each worker $\mathcal{P}_{j}$ for $j\in[\mathcal{V}]$ receives the evaluation $\mathbf{Y}_{i}$ , where $i \in \mathcal{\mathcal{S}^{*}}$. In particular, solving Problem \ref{opt:joint opt} yields a set of evaluation indices that simultaneously minimizes both $\eta_{c}$ and $\bar{P}_{{error}}$. However, to obtain the optimal set $\mathcal{S}^{*}$, one may perform a brute-force search over all possible index combinations, 
which requires evaluating  $\binom{N}{\nu}\!\left(\binom{\nu}{t} + \binom{\nu}{A}(\nu - A)\right)$ MIS and average PEP computations to determine the optimal assignment of evaluation indices for the unreliable worker set $\mathcal{V}$. However, as the system size grows, this approach becomes computationally expensive. Therefore, we propose a low-complexity greedy based algorithm that finds a suboptimal index set, denoted by $\mathcal{S}^{*}_{g}$, which approximates $\mathcal{S}^{*}$. The algorithm begins by selecting an initial subset of size $m\!\!=\!\!\max(A,t)$ that minimizes the objective in Problem \ref{opt:joint opt}. Subsequently, it greedily expands this set by iteratively adding the workers from the remaining pool that minimizes the weighted objective in Problem \ref{opt:joint opt} until the target set size $\nu$ is achieved. We formally present this approach in Algorithm \ref{algo:1}. Further, we compare the computational complexity of the brute-force and greedy approaches used to obtain $\mathcal{S}^{*}$ and $\mathcal{S}^{*}_{g}$, respectively, as shown in Table \ref{tab:complexity_analysis}. Here, complexity is measured in terms of the number of MIS leakage $\eta_{c}$ and average PEP evaluations required.

In Table \ref{tab:variance_sets}, we present the sets $\mathcal{S}^{*}_{g}$ obtained for different $\sigma_p^2$ by solving Problem \ref{opt:joint opt} using our {greedy} based algorithm for $0\leq w\leq 1$, where $w=0$ yields $\mathcal{Z}^{*}$ (minimizing $\bar{P}_{{error}}$) and $w=1$ yields $\mathcal{Q}^{*}$ (minimizing $\eta_c$). The experimental parameters are also listed. We also present simulation plots of $\bar{P}_{{error}}$ and $\eta_{c}$ in Fig. \ref{fig:four_in_row}, obtained by solving Problem \ref{opt:joint opt} at $w = 0.6$ under varying $\sigma_{p}^{2}$. As shown in the figure, choosing $0 < w < 1$ yields a joint solution that simultaneously reduces both $\eta_{c}$ and $\bar{P}_{{error}}$, compared to the extreme cases when $w=0$ (corresponding to $\mathcal{Z}^{*}$) and $w=1$ (corresponding to $\mathcal{Q}^{*}$). Further, for smaller system sizes, we compare the optimal set $\mathcal{S}^{*}$ obtained via brute-force search with the suboptimal set $\mathcal{S}^{*}_{g}$ produced by the {greedy} algorithm under varying $\sigma_{p}^{2}$. The results, shown in Table~\ref{tab:bf_greedy_variance}, indicate that $\mathcal{S}^{*}$ and $\mathcal{S}^{*}_{g}$ are very identical, demonstrating that the {greedy} approach closely approximates the optimal solution  $\mathcal{S}^{*}$.

\begin{table}[t]
    \centering
    \caption{Complexity of brute-force and {greedy} approach.}
    \label{tab:complexity_analysis}
    \renewcommand{\arraystretch}{1.5}
    \resizebox{\columnwidth}{!}{%
    \begin{tabular}{|l|l|l|}
    \hline
    {Method} & {Number of MIS/average PEP computations} & {Complexity Order} \\ 
    \hline
    
    Brute-force & 
    $\binom{N}{\nu}\!\left(\binom{\nu}{t}+\binom{\nu}{A}(\nu-A)\right)$ &
    Exponential \\ 
    \hline
    
    \textit{Greedy} based & 
    $\binom{N}{m} + \sum_{u=m}^{\nu-1}(N-u)
    \big(\binom{u+1}{t}+\binom{u+1}{A}\big)$ &
    Polynomial \\ 
    \hline
    \end{tabular}
}
\end{table}

\section{Summary}
This paper presents a robust NS-LCC framework that accounts for profiled workers. We analyze privacy and accuracy using MIS leakage and localization error, and show that optimizing these two metrics separately yields different assignment sets for unreliable workers revealing a privacy and robustness trade-off. To address this, we formulate a joint optimization problem and develop a low-complexity greedy algorithm that assigns encoded shares to unreliable workers while simultaneously balancing both metrics.


\end{document}